\documentclass[a4paper,12pt]{article}
\usepackage[latin1]{inputenc}
\usepackage{amscd}
\usepackage{nicefrac}
\usepackage{latexsym}
\usepackage{xcolor}
\usepackage[english]{babel}
\usepackage{amsmath,amsthm}
\usepackage{a4wide}
\usepackage{amsfonts}
\usepackage{amssymb}
\usepackage{graphicx}
\usepackage{amsfonts}

\newtheorem{theorem}{Theorem}[section]

\newtheorem{proposition}[theorem]{Proposition}

\theoremstyle{remark}
\newtheorem{rem}{Remark}

\newcommand{\rH}{\mathcal{H}}

\font\timesept=cmr7

\newcommand{\CC}{\mathbb{C}}

\newcommand{\NN}{\mathbb{N}}

\newcommand{\RR}{\mathbb{R}}

\def\ol{\overline}

\newcommand\unit{\hbox{\rm 1\kern-2.8truept l}}

\def\[{{\mathord{[\![}}}
\def\]{{\mathord{]\!]}}}

\def\b{\beta}

\def\td{\widetilde}

\begin{document}

\title{Entanglement of Bipartite Quantum Systems\\driven by Repeated Interactions\footnote{Work supported by ANR project ``HAM-MARK" N${}^\circ$ ANR-09-BLAN-0098-01}}

\author{S. Attal${}^{{}^1}$, J. Deschamps${}^{{}^2}$ and  C. Pellegrini${}^{{}^3}$}

\date{}

\maketitle

\centerline{\timesept ${}^{{}^1}$ Universit\'e de Lyon}
\vskip -1mm
\centerline{\timesept Universit\'e de Lyon 1, C.N.R.S.}
\vskip -1mm
\centerline{\timesept Institut Camille Jordan}
\vskip -1mm
\centerline{\timesept 21 av Claude Bernard}
\vskip -1mm
\centerline{\timesept 69622 Villeubanne cedex, France}
\vskip 2mm
\centerline{\timesept ${}^{{}^2}$ Universit\`a degli Studi di Genova}
\vskip -1mm
\centerline{\timesept Dipartimento di Matematica}
\vskip -1mm
\centerline{\timesept Via Dodecaneso 35}
\vskip -1mm
\centerline{\timesept 16146 Genova, Italy}
\vskip 2mm
\centerline{\timesept ${}^{{}^3}$  Institut de Math\'ematiques de Toulouse }
\vskip -1mm
\centerline{\timesept Laboratoire de Statistique et de Probabilit\'e}
\vskip -1mm
\centerline{\timesept Universit\'e Paul Sabatier (Toulouse III)}
\vskip -1mm
\centerline{\timesept 31062 Toulouse Cedex 9, France}

\begin{abstract}
We consider a non-interacting bipartite quantum system $\mathcal H_S^A\otimes\mathcal H_S^B$ undergoing repeated quantum interactions with an environment modeled by a chain of independant quantum systems interacting one after the other with the bipartite system. The interactions are made so that the pieces of environment interact first with $\mathcal H_S^A$ and then with $\mathcal H_S^B$. Even though the bipartite systems are not interacting, the interactions with the environment create an entanglement. We show that, in the limit of short interaction times, the environment creates an effective interaction Hamiltonian between the two systems. This interaction Hamiltonian is explicitly computed and we show that it keeps track of the order of the successive interactions with $\mathcal H_S^A$ and $\mathcal H_S^B$.
Particular physical models are studied, where the evolution of the entanglement can be explicitly computed.  We also show the property of return of equilibrium and thermalization for a family of examples.
\end{abstract}

\section{Introduction}
Initially introduced in \cite{A-P} in order to justify the quantum Langevin equations, Quantum Repeated Interaction models are currently a very active line of research. They have found various applications: quantum trajectories \cite{P1,P2,P3,Bauer1,Bauer2}, thermalization of quantum systems \cite{DB-P, A-J}. Moreover several famous physical experiments, such as the ones performed by S. Haroche's team, correspond exactly to Quantum Repeated Interaction schemes (\cite{Har, Har2}).

Repeated Quantum Interactions are particular discrete time evolutions of Open Quantum Systems where the typical situation is the one of a quantum system $\mathcal H_S$ in contact with an infinite chain of quantum systems $\bigotimes_k\mathcal H_k$. Each quantum system $\mathcal H_k$ interacts with $\mathcal H_S$ one after the other during a time duration $h$. More concretely, $\mathcal H_1$ interacts with $\mathcal H_S$ during a time duration $h$ and then stops interacting, the second quantum system $\mathcal H_2$ then interacts with $\mathcal H_S$ and so on. The continuous time limit, when $h$ goes to zero, has been studied in detail in \cite{A-P}. Remarkably, it has been shown that such discrete time models, under suitable renormalization, converge to the quantum Langevin equations,  that is, quantum stochastic differential equations.

\smallskip
In this article, we concentrate on the following particular situation. We consider that the system $\mathcal H_S$ is composed of two quantum systems $\mathcal H^A_S$ and $\mathcal H^B_S$ which do not interact together. This ``uncoupled'' system undergoes Quantum Repeated Interactions as follows. Each piece $\mathcal H_k$ of the environment interacts first with $\mathcal H^A_S$ during a time duration $h$ without interacting with $\mathcal H^B_S$  and then interacts with $\mathcal H^B_S$ without interacting anymore with $\mathcal H^A_S$. For example, in the spirit of the experiments driven by Haroche et al (cf \cite{Har, Har2}), the bipartite system can been thought of as two isolated cavities with a magnetic field trapping several photons in each cavities. A chain of two-level systems (such as Rydberg atoms in some particular state, as in the experiment) are passing through the cavities, one after the other, creating this way an entanglement in between the photons of each cavities.

\smallskip
Our work is motivated by entanglement considerations. While the systems $\mathcal H^A_S$ and $\mathcal H^B_S$ are not initially entangled and while there is no direct interaction between them, our special scheme of Quantum Repeated Interactions creates naturally entanglement. More precisely, we show that this scheme of interaction, in the continuous-time limit, is equivalent to a usual Quantum Repeated Interaction model where, actually, $\mathcal H^A_S$ interacts with $\mathcal H^B_S$. In other words, our special scheme of Quantum Repeated Interactions creates spontaneously an effective interaction Hamiltonian between $\mathcal H^A_S$ and $\mathcal H^B_S$. We explicitly compute the associated interaction Hamiltonian.

\smallskip
The article is structured as follows. In Section 2, the bipartite Repeated Quantum Interaction model is described in details. In Section 3, we focus on the continuous-time limit, that is, when the time interaction between the systems $\mathcal H_k$ and $\mathcal H_S=\mathcal H_S^A\otimes\mathcal H_S^B$ goes to zero. More precisely, we derive the quantum stochastic differential equation representing the limit evolution. This allows to identify the effective coupling Hamiltonian. Section 4 is devoted to the study of the evolution of the entanglement between $\mathcal H^A_S$ and $\mathcal H^B_S$ in the physical example of the spontaneous emission of a photon. In Section 5, we derive the Lindblad generator of the limit evolution in the case of a thermal environment, represented by a Gibbs state. We then study the property of return to equilibrium, that is, the asymptotic convergence for all initial state toward an invariant state.

\section{Description of the Bipartite Model}

This section is devoted to the presentation of the model. As announced, we consider a quantum system $\mathcal H_S=\mathcal H^A_S\otimes\mathcal H^B_S$, where $\mathcal H^A_S$ and $\mathcal H^B_S$ do not interact together. This means that the free evolution of $\mathcal H_S$ is given by
$$H^A\otimes I+I\otimes H^B,$$
where $H^A$ and $H^B$ are the free Hamiltonian of $\mathcal H_S^A$ and $\mathcal H_S^B$. This system is coupled to an environment made of an infinite chain of identical and independent systems : 
$$T \Phi=\bigotimes_{k \in \NN^*} \mathcal H_k,$$
where $ \mathcal H_k= \mathcal H$ for all $k$.

The interaction between $\mathcal H_S$ and the infinite chain is described by a model of Quantum Repeated Interactions, that is, the copies of $\mathcal H$ interact ones after the others with $\mathcal H_S$ and then stop interacting. A single interaction between a copy of $\mathcal H$ and $\mathcal H_S=\mathcal H^A_S\otimes\mathcal H^B_S$ is described by a particular mechanism, the interaction is divided into two parts: the system $\mathcal H$ interacts first with $\mathcal H^A_S$ during a time $h$ without interacting with $\mathcal H_S^B$, then the system $\mathcal H$ interacts with $\mathcal H^B_S$ during a time $h$ without interacting with $\mathcal H_S^A$.

In terms of Hamiltonians, the evolution of the coupled system $\mathcal H^A_S\otimes\mathcal H^B_S\otimes\mathcal H$ can be described as follows. For the first interaction, we consider a Hamiltonian of the form
\begin{equation}\label{eq:unitaireA}
H^A_{tot}=H^A\otimes I\otimes I+I\otimes I\otimes H^R+\lambda\, H_{I}^A,
\end{equation}
where $H_R$ represents the free Hamiltonian of $\mathcal H$, the operator $H_{I}^A$ represents the interaction Hamiltonian between $\mathcal H$ and $\mathcal H^A_S$ (this operator acts as the identity operator on $\mathcal H_S^B$) and $\lambda$ is a coupling constant. In a similar way, the second evolution is described by
\begin{equation}\label{eq:unitaireB}
H^B_{tot}=I\otimes H^B\otimes I+I\otimes I\otimes H^R+\lambda'\, H_{I}^B,
\end{equation}
where this time $H_{I}^B$ acts non-trivially only on $\mathcal H$ and $\mathcal H_S^B$ and acts as the identity operator on $\mathcal H_S^A$. The terms $\lambda'$ represents also the coupling constant of the second interaction. Each operator $H^A_{tot}$ and $H^B_{tot}$ gives rise to a unitary evolution
\begin{eqnarray}
U^A&=&e^{-ih H^A_{tot}},\qquad
U^B=e^{-ih H^B_{tot}}.
\end{eqnarray}
Since the space $\mathcal H$ interacts first with $\mathcal H_S^A$ and then $\mathcal H_S^B$, the composed evolution is then given by
\begin{equation}\label{eq:unitairebipartite}
U=U^BU^A.
\end{equation}

\smallskip
Let us stress that usually the interaction between $\mathcal H$ and $\mathcal H_S$ should have been given by a Hamiltonian of the form $\widetilde H_{tot}=H^A\otimes I\otimes I+I\otimes H^B\otimes I+ \,I\otimes I\otimes H^R+\widetilde\lambda\,\widetilde{H}_I,$ where $\widetilde{H}_I$ would have represented the interaction Hamiltonian. This would have given rise to a usual unitary operator of the form 
\begin{equation}\label{eq:unitaireusuel}
\widetilde U=e^{-i2h\widetilde H_{tot}}.
\end{equation} 
In our model, since $H_{I}^A$ and $H_{I}^B$ do not commute, we cannot directly put the unitary \eqref{eq:unitairebipartite} in the form of \eqref{eq:unitaireusuel}. In the following, we shall see (in the continuous-time limit) that our model with $U=U^BU^A$ is equivalent to a standard model where actually there is an effective interaction between $\mathcal H_S^A$ and $\mathcal H_S^B$. 

\smallskip
Let us make precise now the form of the interaction Hamiltonians involved in \eqref{eq:unitaireA} and \eqref{eq:unitaireB}.
We assume in this work that the Hilbert spaces involved in the interaction, that is, $\mathcal H$, $\mathcal H_S^A$ and $\mathcal H_S^B$ are finite dimensional. The Hilbert space $\mathcal H$ is supposed here to be $\CC^{N+1}$. We consider on $\CC^{N+1}$ an orthonormal basis made of eigenvectors of $H^R$, denoted by $\left\{e_0, e_1, \ldots, e_N \right\}$, where the vector $e_0$ is understood as the ground state of $\mathcal H$. 
Consider the associated canonical operators $a^i_j$ defined by 
$$ 
a^i_j e_k = \delta_{ik} e_j\,,
$$
for all $i, j$ and $k$ in $\left\{0, ..., N \right\}$. 
With this notation, we have
$$
H^R= \sum_{j=0}^{N} \lambda_j\, a^j_j\,,
$$
where the $\lambda_j$'s are the eigenvalues of $H^R$. 

For the interaction Hamiltonians we shall consider operators of the form
\begin{eqnarray}
H^A_I&=& \sum_{j=1}^N V_j \otimes I \otimes a^0_j + V_j^* \otimes I \otimes a^j_0\,,\\H^B_I&=& \sum_{j=1}^N I \otimes W_j \otimes a^0_j + I \otimes W_j^* \otimes a^j_0\,,
\end{eqnarray}
where the $V_j$'s are operators on $\mathcal H^A_S$ and the $W_j$'s on $\mathcal H^B_S$.

As usual in the Schr\"odinger picture, the evolutions of states on $\mathcal H^A_S\otimes\mathcal H^B_S\otimes\mathcal H$ are given by
$$\rho \longmapsto U \rho\ U^*\,,$$
where we recall that $U$ takes the particular form $U=U^BU^A$ in our context.

\smallskip
Now, we are in the position to describe the whole interaction between $\mathcal H^A_S\otimes\mathcal H^B_S$ and the chain $\bigotimes_k\mathcal H_k$, with $\mathcal H_k=\mathcal H=\CC^{N+1}$. The scheme is as follows. The first copy $\mathcal H_1$ interacts with $\mathcal H^A_S\otimes\mathcal H^B_S$ during a time $2h$ while the rest of the chain remains isolated. Then, the first copy disappears and the second copy comes to interact and so on... Before making precise the evolution, we need to introduce a notation for the operators acting only on $\mathcal H_n$ and being the identity operator on the rest of the whole space. If $A$ is an operator on $\mathcal H$, we extend it as an operator on $\bigotimes_k\mathcal H_k$ but acting  non-trivially only on $\mathcal H_n$ by putting
$$
A(n)=\bigotimes_{k=1}^{n-1}I\otimes A\otimes \bigotimes_{k>n+1}I\,.
$$
Mathematically, on $\mathcal H^A_S\otimes\mathcal H^B_S\bigotimes_k\mathcal H_k$, we consider the family of unitary operators $(U_n)_{n \in \NN^*}$,  where $U_n$ acts as $U$ on $\mathcal H^A_S \otimes \mathcal H^B_S$ and the $n$-th copy of $\rH$  and as the identity on the rest of the chain. The operator $U_n$ represents actually the interaction between $\mathcal H^A_S\otimes\mathcal H^B_S$ and $\mathcal H_n$. More precisely, the operator $U_n$ is defined as $U_n=U^B_nU^A_n$, where for example $U^A_n=e^{-ih H^A_{tot,n}}$ with
\begin{equation}
H^A_{tot,n}=H^A\otimes I\otimes I+I\otimes I\otimes H_R(n)+\lambda\sum_{j=1}^N V_j \otimes I \otimes a^0_j(n) + V_j^* \otimes I \otimes a^j_0(n)\,.
\end{equation}
The corresponding similar description holds for $U^B_n$. 

The whole evolution is finally described by the family of unitary operators $(V_n)_{n \in \NN^*}$ defined by
\begin{equation}\label{eq:suiteunitaire}
V_n= U_n U_{n-1} \dots U_1\,.
\end{equation}
As a consequence, if the initial state of $\mathcal H^A_S\otimes\mathcal H^B_S\bigotimes_k\mathcal H_k$ is $\rho_0$, the state after $n$ interactions is
$$\rho_0 \longmapsto V_n \,\rho_0\, V_n^*.$$ 
Now that the discrete-time evolution is clearly described, we shall now investigate the continuous-time limit of such evolutions in the next section.

\section{Effective Interaction Hamiltonian}

This section is devoted to derive the continuous time limit of our special scheme of repeated interactions, i.e. the limit when the time parameter $h$ goes to $0$. In order to obtain a relevant limit, the authors of \cite{A-P} have shown that the total Hamiltonian has to be properly rescaled in terms of $h$. In particular, it is crucial to strengthen the interaction in order to see its effect at the limit. More precisely, translated in our context, the total Hamiltonians have to be of the following form:
\begin{eqnarray}
H^A_{tot}&=& H^A \otimes I \otimes I+ I \otimes I \otimes H^R +\dfrac{1}{\sqrt h} \sum_{i=1}^N \left(V_j \otimes I \otimes a^0_j + V_j^* \otimes I \otimes a^j_0\right)\,,\\ 
H^B_{tot}&=& I \otimes H^B \otimes I+ I \otimes I \otimes H^R + \dfrac{1}{\sqrt h}\sum_{j=1}^N I\left( \otimes W_j \otimes a^0_j + I \otimes W_j^* \otimes a^j_0\right)\,.
\end{eqnarray}
Let us stress that in the above expressions the coupling constants appearing in \eqref{eq:unitaireA} and \eqref{eq:unitaireB} have been replaced by  $1/\sqrt{h}$. We denote by $\lfloor\,\cdot\,\rfloor$ the floor function. One can show that the operators $(V_{\lfloor t/h\rfloor})_t$  defined in \eqref{eq:suiteunitaire} converge to a family of operators $(U_t)_t$ satisfying a particular quantum stochastic differential equation. More precisely, in \cite{A-P}, it is shown that we can embed the space $T\Phi$ into an appropriate Fock space $\Phi$. In this context the discrete time interaction, described by $(V_{\lfloor t/h\rfloor})_t$, appears naturally as an approximation of a continuous one described by a family of unitary operators $(U_t)_t$ acting on $\Phi$ (i.e. $(V_{\lfloor t/h\rfloor})_t$ converges to $(U_t)$ when $h$ goes to zero). In particular $(U_t)$ is the solution of a particular \textit{quantum stochastic differential equation} describing continuous-time interaction between small system $\mathcal H_S$ and the quantum field $\Phi$. In our context, the complete description of the Fock space $\Phi$ and the detail of the convergence result are not necessary.  Nevertheless the ``created" interaction Hamiltonian appears naturally in the expression of the limit $(U_t)$.  In our setup, this can be expressed as follows.

\begin{theorem}\label{main}
When the interaction time $h$ goes to $0$, the family $(V_{\lfloor t/h\rfloor})_t$ converges strongly to a family of unitary operators $(U_t)$ which is the solution of the quantum stochastic differential equation
\begin{multline}\label{QSDE}
dU_t =  \hfill\\
\hfill =\left[-i (H^{A}\otimes I + I \otimes H^{B} + 2 \lambda_0 \,I \otimes I) - \dfrac12  \sum_j V_j^* V_j \otimes I + I \otimes W_j^* W_j+ 2 V_j \otimes W_j^* \right] U_t \,dt \\
\hfill-i \sum_{i=1}^N (V_j \otimes I + I \otimes W_j) U_t \ d a^0_j (t) + (V_j^* \otimes I + I \otimes W_j^*) U_t\  d a^j_0(t)\,.
\end{multline}
\end{theorem}

Let us stress that in the expression \eqref{QSDE}, the terms $(a^0_j (t))$ and $(a^j_0(t))$ are quantum noises. They are operators on the limited associated Fock space $\Phi$. The exact definition of these operators is not needed and we refer to \cite{Mey} for complete references.

\begin{proof} In order to prove this result we shall apply the Theorem 13 of \cite{A-P}. The essential step is to identify the relevant terms when we expand 
$$
U= U^B U^A=e^{-ihH^B_{tot}}\,e^{-ihH^A_{tot}}\,,
$$
in terms of $h$.  More precisely, on $\mathcal H^A_S\otimes\mathcal H^B_S\otimes\mathcal H$, one can decompose $U$ as 
\begin{equation}\label{decompose}
U=\sum_{i,j}U_j^i(h)\otimes a_j^i,
\end{equation}  
where $U_j^i(h)$ are operators on $\mathcal H^A_S\otimes\mathcal H^B_S$. This way, we shall find the asymptotic expression of $U_j^i(h)$ and next apply the convergence results of \cite{A-P} to derive Eq. \eqref{QSDE}. 

In order to obtain the asymptotic expression of $U_j^i(h)$, let us study $H^A_{tot}$ and $H^B_{tot}$ in details. Using a similar decomposition as \eqref{decompose}, the operators $H^A_{tot}$ and $H^B_{tot}$ can be seen as matrices whose the coefficients are operators on $\mathcal H^A_S\otimes\mathcal H^B_S$. In particular, they can be written as follows
$$
H^A_{tot}=\small{\left(\begin{array}{ccccc}
H^A\otimes I + \lambda_0 I \otimes I & \frac1{\sqrt h} V_1^* \otimes I &  \frac1{\sqrt h} V_2^* \otimes I & \cdots & \frac1{\sqrt h} V_N^*\otimes I\\
\frac1{\sqrt h} V_1\otimes I& H^A\otimes I + \lambda_1 I \otimes I & 0 & \cdots & 0\\
\frac1{\sqrt h} V_2\otimes I&0 & H^A\otimes I + \lambda_2 I \otimes I &  \ddots & 0\\
\vdots &\vdots & \ddots &\ddots & \vdots\\
\frac1{\sqrt h} V_N\otimes I&0 &0 &\cdots & H^A\otimes I + \lambda_N I \otimes I
\end{array}\right)} 
$$
and
$$ 
H^B_{tot}=\small{\left(\begin{array}{ccccc}
I \otimes H^B+ \lambda_0 I \otimes I & \frac1{\sqrt h} I \otimes W_1^*  &  \frac1{\sqrt h} I \otimes W_2^* & \cdots & \frac1{\sqrt h} I \otimes W_N^*\\
\frac1{\sqrt h} I \otimes W_1& I \otimes H^B + \lambda_1 I \otimes I & 0 & \cdots & 0\\
\frac1{\sqrt h} I \otimes W_2&0 & I \otimes H^B + \lambda_2 I \otimes I &  \ddots & 0\\
\vdots &\vdots & \ddots &\ddots & \vdots\\
\frac1{\sqrt h} I \otimes W_N&0 &0 &\cdots &I \otimes H^B + \lambda_N I \otimes I
\end{array}\right)}\,.
$$
Now, we can make the Taylor expansion of $U$.  We obtain
$$
U^A=\small{\left(\begin{array}{ccccc}
D^A _0 & -i\sqrt h V_1^* \otimes I& -i\sqrt h V_2^* \otimes I&  \cdots & -i \sqrt h V_N^*\otimes I\\
-i\sqrt h V_1\otimes I&D^A_1 & -\frac12 h V_1 V_2^*\otimes I& \cdots & -\frac12 h V_1 V_N^* \otimes I \\
\vdots &-\frac12 h V_2 V_1^*\otimes I & D^A_2 & \ddots & \vdots \\
\vdots & \vdots &\ddots & \ddots &\vdots\\
-i\sqrt h V_N\otimes I & -\frac12 h V_N V_1^* \otimes I & -\frac12 h V_N V_2^* \otimes I  & \cdots& D^A_N 
\end{array}\right)}+ O(h^{3/2}) 
$$
and
$$
U^B=\small{\left(\begin{array}{ccccc}
D^B_0 & -i \sqrt h  I \otimes W_1^*& -i \sqrt h  I \otimes W_2^*&  \cdots & -i \sqrt h  I \otimes W_N^*\\
-i\sqrt h I \otimes W_1&D^B_1 &  -\frac12 h I \otimes W_1 W_2^*& \cdots &  -\frac12 h I \otimes W_1 W_N^*\\
\vdots & -\frac12 h I \otimes W_2 W_1^*& D^B_2 & \ddots & \vdots \\
\vdots & \vdots &\ddots & \ddots &\vdots\\
-i\sqrt h I \otimes W_N&  -\frac12 h I \otimes W_N W_1^*  &  -\frac12 h I \otimes W_N W_2^* & \cdots& D^B_N 
\end{array}\right)} + O(h^{3/2})
$$
where the diagonal coefficients of $U^A$ and $U^B$ are for all $j=1, \ldots, N$,
\begin{align*}
D^A_0 &= I \otimes I -ih H^A\otimes I -ih \lambda_0 I \otimes I -\frac12 h\sum_j V_j^* V_j \otimes I\,,\\
D^A_j &= I \otimes I -ih H^A\otimes I -ih \lambda_j I \otimes I  - \frac12 h V_j V_j^* \otimes I\,,\\
D^B_0 &= I \otimes I -ih I \otimes H^B -ih \lambda_0 I \otimes I-\frac12 h\sum_j I \otimes  W_j^* W_j\,,\\
D^B_j &= I \otimes I -ih I \otimes H^B   -ih \lambda_j I \otimes I- \frac12 h I \otimes W_j W_j^*\,.
\end{align*}
This way, computing $U^BU^A$ in asymptotic form, the coefficients $U^i_j(h)$ of the matrix $U$ for $i,j=0, \ldots, N$ are, up to terms in $h^{3/2}$ or higher
\begin{align}
U^0_0&= I\otimes I -ih(H_{A}\otimes I + I \otimes H_{B} + 2 \lambda_0 I \otimes I)\nonumber\\
&\hphantom{ I\otimes I -ih(H_{A}\otimes I + I \otimes H_{B} + }- \frac12 h  \sum_j \left(V_j^* V_j \otimes I + I \otimes W_j^* W_j+ 2 V_j \otimes W_j^*\right)\,,\label{uh1}\\
U_0^j&=-i\sqrt h(V_j^* \otimes I + I \otimes W_j^*) \,,\label{uh2}\\
U_j^0&=-i\sqrt h(V_j\otimes I + I \otimes W_j) \,,\label{uh3}\\
U_j^j&=I \otimes I -ih(H_{A}\otimes I + I \otimes H_{B} + 2 \lambda_j I \otimes I) \nonumber\\
&\hphantom{I\otimes I -ih(H_{A}\otimes I + I \otimes H_{B} + }-\frac12 h (V_jV_j^*\otimes I +I\otimes W_j W_j^*+2V_j^* \otimes W_j )\,,\label{uh4}\\
U_j^k&= -\frac12 h( V_j V_k^*\otimes I + I \otimes W_j W_k^*+2V_k^*\otimes W_j)\,.\label{uh5}
\end{align}
We can easily check that
$$
\lim_{h \rightarrow 0} \dfrac{U^i_j(h)-\delta_{ij} \,I\otimes I}{h^{\epsilon_{i,j}}}=L^i_j\,,
$$
where $\epsilon_{0,0}=1$, $\epsilon_{0,j}= \epsilon_{j,0}= 1/2$ and $\epsilon_{i,j}=0$ and where
\begin{align*}
L^0_0&= -i (H^{A}\otimes I + I \otimes H^{B} + 2 \lambda_0 I \otimes I) - \frac12  \sum_j V_j^* V_j \otimes I + I \otimes W_j^* W_j+ 2 V_j \otimes W_j^*\,,\\
L^j_0&= -i(V_j^* \otimes I + I \otimes W_j^*) \,,\\
L^0_j&= -i(V_j\otimes I + I \otimes W_j)\,,\\
L^i_j&=0.
\end{align*}
These are exactly the conditions of \cite{A-P} and the result follows.
\end{proof}

\smallskip
Now that we have derived Eq. \eqref{QSDE}, we are in the position to identify the interaction Hamiltonian which has been ``created" by the environment. To this end, we compare the limit equation \eqref{QSDE} with the one one could have obtained in a usual repeated quantum interaction scheme. 

\begin{theorem}\label{main2}
The quantum stochastic differential equation \eqref{QSDE} represents an evolution, on $\mathcal H_S^A\otimes \mathcal H_S^B$ coupled to a Fock space $\Phi$, which can be obtained from the continuous-time limit of a usual repeated interaction scheme with the following Hamiltonian on $\mathcal H_S^A\otimes \mathcal H_S^B\otimes \rH$
\begin{equation}\label{usualtotalhamiltonian}
\widetilde H_{tot}= H_0^{A,B}\otimes I+2 \,I\otimes I\otimes H^R+ \frac{1}{\sqrt{h}}\sum_j S_j \otimes a^0_j + S_j^* \otimes a^j_0\,, 
\end{equation}
where $S_j=V_j\otimes I + I \otimes W_j$ and where the free Hamiltonian of $\mathcal H_S^A\otimes \mathcal H_S^B$ is given by
\begin{equation}\label{usualfreehamiltonian}
H_0^{A,B}=H^{A}\otimes I + I \otimes H^{B} +\dfrac i2 \sum_j V_j^* \otimes W_j - V_j \otimes W_j^*\,.
\end{equation}

In particular the term 
$$
\frac{i}2\, \sum_j \left(V_j^* \otimes W_j - V_j \otimes W_j^*\right)
$$ 
represents an effective interaction Hamiltonian term created by the environment between $\mathcal H_S^A$ and $\mathcal H_S^B$.
\end{theorem}

\begin{proof}
With the expression of the Hamiltonian \eqref{usualtotalhamiltonian}, using again the results of  \cite{A-P} the continuous-time limit, $h$ goes to zero, gives rise to the QSDE
\begin{equation}\label{usualQSDE}
d\td U_t=L_0^0\td U_tdt+\sum_j L_0^j\td U_tda_0^j(t)+L_j^0\td U_tda_j^0(t)\,,
\end{equation}
where 
\begin{align*}
L^0_0&= -i( H_0^{A,B} + 2\lambda_0 I\otimes I)-\dfrac12\sum_j S_j^* S_j,\\
L_j^0&=- i S_j\qquad \text{and}\qquad
L_0^j= -i S_j^*\,.
\end{align*}
In Eq. \eqref{QSDE}, the coefficient $L^0_0$ is 
$$L^0_0= -i (H^{A}\otimes I + I \otimes H^{B} + 2 \lambda_0 I \otimes I) - \frac12  \sum_j V_j^* V_j \otimes I + I \otimes W_j^* W_j+ 2 V_j \otimes W_j^*\,$$
which can be also written
\begin{align*}
L^0_0=& -i\left( H^{A}\otimes I + I \otimes H^{B} + 2\lambda_0 I \otimes I +\dfrac i2 \sum_j V_j^* \otimes W_j - V_j \otimes W_j^*\right)\\&- \dfrac12\sum_j (V_j\otimes I + I \otimes W_j)^* (V_j\otimes I + I \otimes W_j)\,.
\end{align*}
Note that finally Eq.\eqref{QSDE} is exactly a QSDE of the form of Eq.\eqref{usualQSDE} with
$$
H_0^{A,B}=H^{A}\otimes I + I \otimes H^{B} +\dfrac i2 \sum_j V_j^* \otimes W_j - V_j \otimes W_j^*\quad \text{ and }\quad S_j=V_j\otimes I + I \otimes W_j\,.
$$
The result follows.
\end{proof}

\smallskip
Before treating an example where we study the entanglement created by the interaction with the environment, we want to highlight the results of the previous theorems with the following remarks. 

\begin{rem} One can wonder if we can recover the above result and the description of the created interaction Hamiltonian only by knowing the separate evolutions (that is when only $\mathcal H^A_S$ or $\mathcal H^B_S$ is involved). Using again the results of \cite{A-P} we can describe the separate evolution and we get
\begin{equation}\label{eq:usuelA}
d\td U_t^A=[-i(H^A+ \lambda_0 I)-\dfrac12\sum_j V_j^* V_j] \td U_t^Adt-i\sum_j V_j \td U_t^Ada_j^0(t)+V_j^*\td U_t^Ada_0^j(t)\,,
\end{equation}
which is the limit of $V_{[t/h]}^A=U_{[t/h]}^A...U^A_1$. The similar evolution of $\mathcal H_S^B$ coupled to the Fock space $\Phi$ is given by
\begin{equation}\label{eq:usuelB}
d\td U_t^B=[-i(H^B+ \lambda_0 I)-\dfrac12\sum_j W_j^* W_j] \td U_t^Bdt-i\sum_j W_j\td U_t^Bda_j^0(t)+W_j^*\td U_t^B\ da_0^j(t)\,,
\end{equation}
which is the limit of $V_{[t/h]}^B=U_{[t/h]}^B...U^B_1$. At this stage, with only \eqref{eq:usuelA} and \eqref{eq:usuelB}, it is not absolutely straightforward how to derive Eq.\eqref{QSDE}. In particular, it is not obvious how to describe the fact that the quantum field, at time $t$, acts first with $\mathcal H^A_S$ and second with $\mathcal H^B_S$

The only way to model such an evolution as the continuous time limit of the usual repeated interactions scheme is to consider that the free Hamiltonian of $\mathcal H^A_S\otimes\mathcal H^B_S$ is of the form
 $$H_0^{A,B}=H^{A}\otimes I + I \otimes H^{B} +\dfrac i2 \sum_j V_j^* \otimes W_j - V_j \otimes W_j^*$$
 which involves then an effective interaction $i/2 \sum_j V_j^* \otimes W_j - V_j \otimes W_j^*$ as proved in Theorem~\ref{main2}. This shows that our special model creates spontaneously an interaction Hamiltonian between $\mathcal H^A_S$ and $\mathcal H^B_S$ since they are supposed to be isolated at the beginning.
 \end{rem}
\smallskip

\begin{rem}
Note that the interaction Hamiltonian is not symmetric in $\mathcal H^A_S$ and $\mathcal H^B_S$ due to the fact that each auxiliary system $\mathcal H$ acts with $\mathcal H^A_S$ before $\mathcal H^B_S$. Somehow the evolution keeps the memory of the order of the interaction.
\end{rem}
\smallskip

\smallskip
 
In the next section, we illustrate our results with the study of the creation of entanglement in a simple physical model.

\section{Evolution of Entanglement for Spontaneous Emission}\label{physicalexample}

The physical model considered in this section is the spontaneous emission of photons. More precisely, the systems $\mathcal H^A_S$, $\mathcal H^B_S$ and $\mathcal H$ are $2$-level systems represented by $\CC^{2}$. The interaction between one of the spaces $\mathcal H^A_S$ or $\mathcal H^B_S$ and $\mathcal H$ is the exchange of energy, that is, the following.

 The free dynamics $H^A$, $H^B$ and $H^R$ are given by  the usual Pauli matrix $\sigma_z$ defined by
$$\left(\begin{array}{cc}
1&0\\
0&-1
\end{array}\right)\,.$$
The operators $V_1$ and $W_1$ are $V_1=W_1=a^1_0$.
Applying Theorem \ref{main}, the limit evolution is 
\begin{multline}\label{emission}
dU_t =  \left[-i ( \sigma_z \otimes I + I \otimes  \sigma_z + 2 I \otimes I) - \dfrac12  S^* S+\dfrac12 (a^0_1 \otimes a^1_0 - a^1_0\otimes a^0_1)\right] U_t dt\\
 \hfill -i  S U_t d a^0_1 (t) -iS^* U_t d a^1_0(t)\,,
\end{multline}
where $S= a^1_0 \otimes I + I \otimes a^1_0$.
\bigskip

In order to study the entanglement of a system evolving according to Eq.\eqref{emission}, we compute its Lindblad generator.
Indeed,  from the solution $(U_t)_{t \in \RR^+}$ of Eq.\eqref{emission}, we consider associate the semigroup of completely positive maps 
$(T_t)_{t \in \RR^+}$ defined by
$$
T_t(\rho)= Tr_{\rH}(U_t(\rho \otimes \vert \Omega \rangle \langle \Omega \vert)U^*_t)\,,
$$
for all state $\rho$ of $\mathcal H^A_S\otimes \mathcal H^B_S$ and all $t \in \RR^+$, where $\Omega$ represents the ground (or vacuum) state of the associated Fock space $\Phi$.  The infinitesimal generator of $(T_t)$ is then given by  
\begin{align*}
L(\rho)= &-i\left[\sigma_z \otimes I + I \otimes \sigma_z +\dfrac{i}2 (a^0_1 \otimes a^1_0 - a^1_0 \otimes a^0_1), \rho \right] + \dfrac12 \Big(2S\rho S^* - S^*S\rho-\rho S^*S\Big)\,.
\end{align*}
Note that this generator can also be simply recovered from the limit of the completely positive discrete-time semigroup associated to the completely positive operator $l(h)$ defined by
\begin{eqnarray}
l(h)(\rho)&=&\textrm{Tr}_\mathcal{H}(U(\rho\otimes\vert e_0\rangle\langle e_0\vert)U)\\
&=&\sum_{i}U^0_i(h)\,\rho\, U^0_i(h)^*\\
&=&\rho+h L(\rho)+\circ(h)\,.
\end{eqnarray}

\smallskip
Now we are in the position to study the entanglement between the system $\mathcal H_S^A$ and $\mathcal H_S^B$. In particular, we shall study the so-called \textit{entanglement of formation}  (see \cite{V} for an introduction). It is worth noticing that an explicit formula does not hold in general. In \cite{V}, an explicit formula has been derived for particular initial states. These initial states are called $X$ states for their matrix representations look like an $X$. A particular feature of such states is that this representation is preserved during the dynamics and the entanglement of formation can be computed explicitly in terms of the concurrence of Wooters \cite{W}.

In order to make concrete the $X$ representation, we consider the following basis of $\mathcal H^A_S \otimes \mathcal H^B_S$:$$\mathcal B=(\vert e_0 \otimes e_0\rangle,\vert e_0 \otimes e_1\rangle,\vert e_1 \otimes e_0\rangle,\vert e_1 \otimes e_1\rangle)\,.$$
A general $X$ state, in this basis, is then
$$\rho = \left( \begin{array}{cccc}
a&0&0&y\\
0&b&x&0\\
0&\overline{x} &c&0\\
\overline{y} &0&0&d
\end{array}\right)$$
 with the conditions $a,b,c,d$ non-negative reals such that $a+b+c+d=1$, $\vert y \vert^2 \leq ad$ and $\vert x \vert^2 \leq bc$. As proved in \cite{V}, the concurrence of Wooters of such a state is
\begin{equation}\label{concurrence}C(\rho)=2 \max(0, \vert y \vert - \sqrt{bc}, \vert x \vert - \sqrt{ad})\,\end{equation}
and its entanglement of formation is given by the general formula, shown by Wooters \cite{W},
\begin{equation}\label{entanglement}
E(\rho)= h\left( \dfrac{1+ \sqrt{1-C(\rho)^2}}2\right)\,,
\end{equation}
where $h(x)= -x \log_2 (x) - (1-x) \log_2( 1-x)$.

One can now compute the action of $L$ on a $X$ state and after computation we get
$$
L(\rho) = \left( \begin{array}{cccc}
\ol x + x + b+ c&0&0&y(-1-4i)\\
0&d-b-\overline{x}-x&d-c-x&0\\
0&d-c-\overline{x} &d-c&0\\
\overline{y}(4i-1) &0&0&-2d
\end{array}\right)\,.
$$
Using the development of $e^{tL}$ in serie, it is obvious to see that the $X$ representation is preserved during the evolution. Unfortunately, in general, the expression of $L^n(\rho)$ for all $n$ is not computable and we cannot obtain the expression of $e^{tL}(\rho)$ for all $\rho$. However, we are able to compute the expression of $e^{tL}(\rho)$ for the states defining the basis $\mathcal B$.

\smallskip\noindent
$\bullet$ A straightforward computation shows that $\vert e_0 \otimes e_0\rangle\langle e_0 \otimes e_0\vert$ is an invariant state of the dynamics (one can check that $L(\vert e_0 \otimes e_0\rangle\langle e_0 \otimes e_0\vert)=0)$ and there is no entanglement of formation.

\smallskip\noindent
$\bullet$ Consider now another initial state $\rho^{01}=\vert e_0 \otimes e_1\rangle\langle e_0 \otimes e_1\vert$ corresponding to the case $a=0$, $ b=1$, $c=0$, $d=0$, $x=0$ and $y=0$. This state represents the system $\mathcal H_S^A$ in its ground state and $\mathcal H_S^B$ in its excited state. One can easily check that we get for all $n\geq 1$
$$L^n(\rho^{01}) = \left( \begin{array}{cccc}
(-1)^{n+1}&0&0&0\\
0&(-1)^n&0&0\\
0&0&0 &0\\
0&0&0&0
\end{array}\right)\,.$$
This gives directly that
\begin{eqnarray}
\rho^{01}_t&=& e^{tL}(\vert e_0 \otimes e_1\rangle\langle e_0 \otimes e_1\vert) = \left( \begin{array}{cccc}
1-e^{-t}&0&0&0\\
0& e^{-t}&0&0\\
0&0 &0&0\\
0 &0&0&0
\end{array}\right)\,\\&=&\vert e_0\otimes (1-e^{-t})e_0+e^{-t}e_1\rangle\langle e_0\otimes (1-e^{-t})e_0+e^{-t}e_1\vert
\end{eqnarray}
The entanglement of formation is obviously zero. Which was to be expected as the initial state of $\mathcal H_S^A$ is $\vert e_0\rangle$  is invariant under the repeated interactions. In the two next cases we shall see effective creation of entanglement.

\smallskip\noindent
\begin{figure}[h!]
\begin{center}
\leavevmode
        \includegraphics[width=10cm,height=7cm]
        {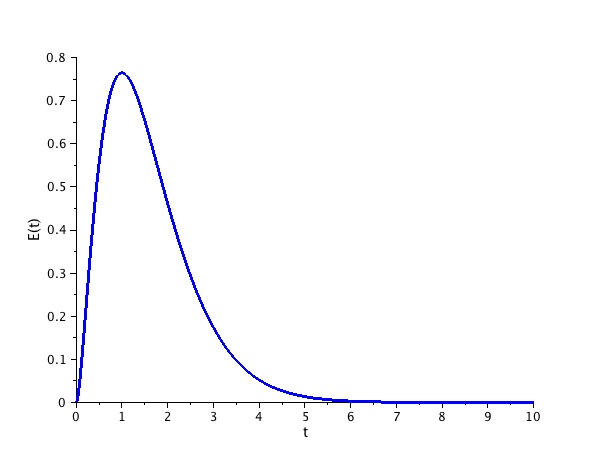}\\
     \vspace{-0.9cm} \strut
\end{center}
\vskip -0.1cm
\caption{\emph{Time evolution of Wooters' concurrence, initial state $e_1 \otimes e_0$}}
\end{figure}

$\bullet$ Consider the initial state $\rho^{10}=\vert e_1 \otimes e_0\rangle\langle e_1 \otimes e_0\vert$ corresponding to the case $a=0$, $b=0$, $c=1$, $d=0$, $x=0$ and $y=0$. For all $n\geq 1$, we get
$$L^n(\rho^{10}) = \left( \begin{array}{cccc}
(-1)^{n+1}(n^2-n+1)&0&0&0\\
0&(-1)^n(n-1)n&(-1)^n n&0\\
0&(-1)^n n&(-1)^n  &0\\
0&0&0&0
\end{array}\right)\,.$$
This way, we have for all time $t$,
\begin{equation}\label{rhot}
\rho^{10}_t= e^{tL}(\vert e_1 \otimes e_0\rangle\langle e_1 \otimes e_0\vert) = \left( \begin{array}{cccc}
1-(1+ t^2) e^{-t}&0&0&0\\
0&t^2 e^{-t}&-t e^{-t}&0\\
0&-te^{-t} &e^{-t}&0\\
0 &0&0&0
\end{array}\right)\,.
\end{equation}
In this case the entanglement of formation is then
\begin{equation}\label{entanglementrhot}
E(\rho^{10}_t)= h\left( \dfrac{1+ \sqrt{1-4 t^2 e^{-2t}}}2\right)\,.
\end{equation}
In particular, this quantity is positive for all $t>0$ (see figure). One can check that the maximum is reached at time $1$ when the state is
$$\left( \begin{array}{cccc}
1-2\ e^{-1}&0&0&0\\
0&e^{-1}&-e^{-1}&0\\
0&-e^{-1} &e^{-1}&0\\
0 &0&0&0
\end{array}\right)\,.$$
In this case we see that there is spontaneous creation of entanglement which increases until time $1$ and next decreases exponentially fast to zero, see Fig. 1.

\smallskip
\begin{figure}[]
\begin{center}
\leavevmode
        \includegraphics[width=10cm,height=7cm]
        {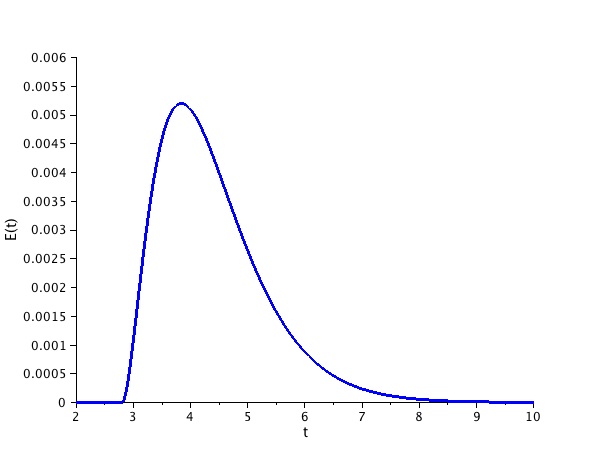}\\
     \vspace{-0.9cm} \strut
\end{center}
\vskip -0.1cm
\caption{\emph{Time evolution of Wooters' concurrence, initial state $e_1 \otimes e_1$}}
\end{figure}
$\bullet$ The last case concerns $\rho^{11}=\vert e_1 \otimes e_1\rangle\langle e_1 \otimes e_1\vert$. In particular, this corresponds to the case $a=0$, $b=0$, $c=0$, $d=1$, $x=0$ and $y=0$. After computations we get 
$$L^n(\rho^{11}) =(-1)^{n} \left( \begin{array}{cccc}
5\times 2^{n} -6-n(n+3)&0&0&0\\
0&-5(2^{n}-1)+n(n+3)&-2^{n+1}+n+2&0\\
0&-2^{n+1}+n+2&-2^n+1&0\\
0&0&0&2^n
\end{array}\right)\,.$$
This gives for all time $t$,
\begin{equation}\label{rhot}
\rho^{11}_t= \left( \begin{array}{cccc}
1-(t^2-4t+6) e^{-t}+5e^{-2t}&0&0&0\\
0&(t^2-4t+5) e^{-t}-5e^{-2t}&(2-t) e^{-t}-2te^{-2t}&0\\
0&(2-t) e^{-t}-2e^{-2t}&e^{-t}-e^{-2t}&0\\
0 &0&0&e^{-2t}
\end{array}\right)\,.
\end{equation}
The concurrence of Wooters is then
 $$C(\rho^{11}_t)= 2 \max [0, \vert (2-t) e^{-t}-2e^{-2t}\vert -\sqrt{(1-(t^2-4t+6) e^{-t}+5e^{-2t})e^{-2t}}]\,.$$
and the entanglement of formation is
\begin{equation}
E(\rho^{11}_t)= h\left( \dfrac{1+ \sqrt{1-C(\rho^{11}_t)^2}}2\right)\,.
\end{equation}
The behavior is mostly the same as in the previous case (see Fig 2), with the important difference that the entanglement, initially starting at 0, takes a strictly positive time to leave the value 0.

\section{Thermal Environment}\label{casthermique}

In this section, we want to investigate the bipartite model in presence of a thermal environment. To this end, we consider that the reference state of each copy of $\mathcal H$ is the Gibbs state 
$$
\omega_{\beta} = \dfrac1{Z}e^{-\beta H^R}\,,
$$
where $\beta$ is positive and $Z$ is a normalizing constant (as usual $\beta$ is the inverse of the temperature). In the orthonormal basis $\left\{e_0, \dots, e_n\right\}$ of  eigenvectors of the Hamiltonian $H^R$, the state $\omega_{\beta}$ is diagonal and is expressed as
\begin{equation}\label{Gibbs}
\omega_{\beta}=\sum_j\beta_j\vert e_j\rangle\langle e_j\vert\,,
\end{equation}
where $\b_j= e^{-\b \lambda_j}/Z$ with $\sum_j \b_j= 1$. 

Let us stress that the limit evolution described in \cite{A-P} is crucially related to the fact that the state of $\mathcal H$ is a pure state. With a general state of the form $\omega_\beta$, in order to compute the limit evolutionin terms of a unitary evolution on a Fock space, one has to consider a so-called of G.N.S. representation of the dynamics. This technics has been successfully developed in \cite{A-J} in order to derive the quantum Langevin equation associated to the action of  a quantum heat bath. 

\subsection{Limit Lindblad Generator}

Here we shall not describe such results and we focus only on the Lindblad generator. As in the previous section this generator can be obtained from the continuous-limit of the discrete one. To this end we define the discrete generator $l_{\beta}(h)$ including temperature  by
\begin{equation}\label{thermal}
l_{\beta}(h)(\rho)=\textrm{Tr}_\mathcal{H}(U(\rho\otimes\omega_{\beta})U)= \sum_{k}\beta_k\textrm{Tr}_\mathcal{H}(U(\rho\otimes\vert e_k\rangle\langle e_k\vert)U) = \sum_{j,k}\beta_k\,  U_{j}^k\, \rho\, U_j^{k*}\,.
\end{equation}
\begin{proposition}\label{thermique}
In terms of $h$, the asymptotic expression of $l_\beta(h)$ is given by
$$l_\beta(h)(\rho)=\rho+hL_\beta(\rho)+\circ(h),$$
where
\begin{align}\label{lindbladthermique}
L_{\beta}(\rho)=&  -i\left[H^A\otimes I + I \otimes H^B  + \dfrac{i}2  \sum_{j=1}^N (\beta_j- \beta_0)  (V_j \otimes W_j^* - V_j^* \otimes W_j), \rho \right]\\
\nonumber&- \dfrac12 \sum_{j=1}^N \beta_j \big(S_jS_j^*\rho + \rho S_jS_j^*- 2 S_j^* \rho S_j \big)- \dfrac12 \sum_{j=1}^N \beta_0 \big(S^*_j S_j\rho+   \rho S^*_j S_j   -2 S_j \rho S_j^*\big)\,,
\end{align}
where $S_j = V_j \otimes I + I \otimes W_j$.

Furthermore, the interaction Hamiltonian between $\mathcal H^A_S$ and $\mathcal H^B_S$ created by repeated interactions with the environment is
$$ \dfrac{i}2  \sum_{j=1}^N (\beta_j- \beta_0)  (V_j \otimes W_j^* - V_j^* \otimes W_j)\,.$$
\end{proposition}
\begin{proof}

Plugging the asymptotic expressions \eqref{uh1} -- \eqref{uh5} into \eqref{thermal} and putting 

\noindent $S_j=V_j\otimes I+I\otimes W_j$ for all $j \geq 1$, we get, up to terms in $h^{3/2}$ or higher,
\begin{align*}
l_{\beta}(\rho)=\rho + h \Big( -i\left[H^A\otimes I + I \otimes H^B , \rho \right] 
&- \dfrac12 \sum_{j=1}^N \beta_j (V_jV_j^*\otimes I +I\otimes W_j W_j^*+2V_j^* \otimes W_j)\rho\\
&- \dfrac12 \sum_{j=1}^N \beta_j \rho (V_jV_j^*\otimes I +I\otimes W_j W_j^*+2V_j^* \otimes W_j)^* \\
&- \dfrac12 \sum_{j=1}^N\beta_0 (V_j^*V_j\otimes I +I\otimes W_j^* W_j+2V_j \otimes W_j^*)\rho \\ 
&- \dfrac12 \sum_{j=1}^N \beta_0 \rho (V_j^*V_j\otimes I +I\otimes W_j^* W_j+2V_j \otimes W_j^*)^*\\
 &+  \sum_{j=1}^N \beta_0 S_j \rho S_j^* + \beta_j S_j^* \rho S_j \Big),
 \end{align*}
 which can be written in the usual form
 \begin{align*}
l_{\beta}(\rho)=& \rho + h \Big( -i\left[H^A\otimes I + I \otimes H^B  + \dfrac{i}2  \sum_{j=1}^N (\beta_j- \beta_0)   (V_j \otimes W_j^* - V_j^* \otimes W_j), \rho \right]\\
&- \dfrac12 \sum_{j=1}^N \beta_j \big(S_jS_j^*\rho + \rho S_jS_j^*- 2 S_j^* \rho S_j \big)- \dfrac12 \sum_{j=1}^N \beta_0 \big(S^*_j S_j\rho+  \beta_0 \rho S^*_j S_j   -2 S_j \rho S_j^*\big) \Big)\,.
\end{align*}
This way, the interaction Hamiltonian naturally appears in the Hamiltonian part.
\end{proof}

In the next section, in a particular example, we study the asymptotic behavior of $\mathcal H_S^A\otimes\mathcal H_S^B$.

\subsection{Return to Equilibrium in a Physical Example, Thermalization}\label{physicalexample}
Here, we consider that  $\mathcal H^A_S$, $\mathcal H^B_S$ and $\rH$ are $\CC^{N+1}$. We assume that the free evolutions satisfy $H^A=H^B=H^R$.
The total Hamiltonian operators are
\begin{eqnarray}H^A_{tot}&=&H^A \otimes I \otimes I + I \otimes I \otimes H^R+\frac{1}{\sqrt h}\sum_{j =1}^N a^0_j \otimes I \otimes a^j_0 + a^j_0 \otimes I \otimes a^0_j\,,\\
H^B_{tot}&=& I \otimes H^B \otimes I +  I \otimes I \otimes H^R+ \frac{1}{\sqrt h}\sum_{j=1}^N I \otimes a^0_j \otimes a^j_0 + I \otimes a^j_0 \otimes a^0_j.
\end{eqnarray}
This is a generalization of the spontaneous emission (see \cite{A-J}).

Applying Proposition \ref{thermique} we get the expression of the Lindblad generator\begin{align*}
L_{\beta}(\rho)=&  -i\left[H^A\otimes I + I \otimes H^B  + \dfrac{i}2  \sum_{j=1}^N (\beta_j- \beta_0)  (a^j_0 \otimes a^0_j - a^0_j \otimes a^j_0), \rho \right]\\
\nonumber&- \dfrac12 \sum_{j=1}^N \beta_j \big(S_jS_j^*\rho + \rho S_jS_j^*- 2 S_j^* \rho S_j \big)- \dfrac12 \sum_{j=1}^N \beta_0 \big(S^*_j S_j\rho+   \rho S^*_j S_j   -2 S_j \rho S_j^*\big)\,,
\end{align*}
where $S_j = a^j_0 \otimes I + I \otimes a^j_0$. 

\smallskip
Now, we are in the position to consider the problem of return to equilibrium. More precisely, we shall show that there exists a unique state $\rho_\infty$ such that
$$
\lim_{t \rightarrow +\infty} \textrm{Tr}(e^{tL_{\beta}}(\rho) X) = \textrm{Tr}(\rho_{\infty} X)\,,
$$
for all initial state $\rho$ and all observable $X$ on $\mathcal H^A_S \otimes \mathcal H^B_S$. The state $\rho_{\infty}$ is an invariant state.

In the case of finite dimensional Hilbert spaces, a general result, proved by Frigerio and Verri \cite{F-V} and extended by Fagnola and Rebolledo \cite{F-R} gives a sufficient condition, in the case where the system has a faithful invariant state $\rho_{\infty}$. The criterion is the following. Let $L$, defined by
$$
L(\rho)=-i[H,\rho]+\sum_j-\frac{1}{2}\{C_j^\star C_j,\rho\}+C_j\rho C_j^\star\,,
$$
be the Lindblad generator of a quantum dynamical system. The property of return to equilibrium is satisfied if
\begin{eqnarray}
\label{comm}\left\{H, L_j, L_j^* ;   j=1, \ldots,N\right\}' = \quad \left\{ L_j, L_j^* ;   j=1, \ldots, N\right\}'\,,
\end{eqnarray}
where the notation $\{\}'$ refers to the commutant of the ensemble.

In our context we shall prove the following return to equilibrium result.

\begin{theorem}\label{equilibrium}
On $\mathcal H^A_S\otimes\mathcal H^B_S$, the dynamical system whose Lindblad generator is given by
\begin{align*}
L_{\beta}(\rho)=&  -i\left[H^A\otimes I + I \otimes H^B  + \dfrac{i}2  \sum_{j=1}^N (\beta_j- \beta_0)  (a^j_0 \otimes a^0_j - a^0_j \otimes a^j_0), \rho \right]\\
\nonumber&- \dfrac12 \sum_{j=1}^N \beta_j \big(S_jS_j^*\rho + \rho S_jS_j^*- 2 S_j^* \rho S_j \big)- \dfrac12 \sum_{j=1}^N \beta_0 \big(S^*_j S_j\rho+   \rho S^*_j S_j   -2 S_j \rho S_j^*\big)\,,
\end{align*}
where $S_j = a^j_0 \otimes I + I \otimes a^j_0$, has the property of return to equilibrium. 

Moreover, the limit invariant state is
$$\rho_{\beta}= \dfrac{e^{-\beta (H^A \otimes I + I \otimes H^B)}}{Z}\,,$$
where $Z$ is a normalizing constant.
\end{theorem}
\begin{proof}
Firstly, one can check that $\rho_{\beta}$ is a faithful invariant state since
$$L_{\beta}(\rho_{\beta})=0\,.$$
The proof is then based on the result of Fagnola and Rebolledo by showing that the commutants
$$
\big\{ H^A\otimes I + I \otimes H^B  + \dfrac{i}2  \sum_{j=1}^N (\beta_j- \beta_0)  (a^j_0 \otimes a^0_j - a^0_j \otimes a^j_0), S_k, S_k^* ;   k=1, \ldots, N\big\}'$$ and $$\left\{ S_k, S_k^* ;   k=1, \ldots, N\right\}'
$$ 
are simply trivial.

Recall that in this physical system the operators $S_k$ are $a^k_0\otimes I + I \otimes a^k_0$ for all $k~\geq~1$. Let us prove now that  $\left\{ S_k, S_k^* ;   k=1, \ldots, N\right\}'$ is trivial.
Consider an element $K$ of this commutant. This element $K$ can be written with respect to the canonical basis $(a^i_j)_{i,j=0, ...,N}$ as
$$K= \sum_{i,j=0}^N K^i_j \otimes a^i_j\,,$$
where the $K^i_j$'s are operators on $\CC^{N+1}$. Since the operators $K$ and $S_k$ commute for all $k\geq 1$, we get the equality between
\begin{align*}
K S_k=  \left(\sum_{i,j=0}^N K^i_j \otimes a^i_j \right) \left( a^k_0\otimes I + I \otimes a^k_0 \right)&= \sum_{i,j=0}^N K^i_j a^k_0 \otimes a^i_j + \sum_{i,j=0}^N K^i_j \otimes a^i_j a^k_0\\
&= \sum_{i,j=0}^N K^i_j a^k_0 \otimes a^i_j + \sum_{j=0}^N K^0_j \otimes  a^k_j
\end{align*}
and
\begin{align*}
S_k K=  \left( a^k_0\otimes I + I \otimes a^k_0 \right) \left(\sum_{i,j=0}^N K^i_j \otimes a^i_j \right)&= \sum_{i,j=0}^N a^k_0 K^i_j \otimes a^i_j + \sum_{i,j=0}^N K^i_j \otimes a^k_0 a^i_j  \\
&=\sum_{i,j=0}^N a^k_0 K^i_j \otimes a^i_j  + \sum_{i=0}^N K^i_k \otimes  a^i_0\,.
\end{align*}
From the commutation of $K$ and $S_k^*$, we also have the equality between
\begin{align*}
K S_k^*=  \left(\sum_{i,j=0}^N K^i_j \otimes a^i_j \right) \left( a^0_k\otimes I + I \otimes a^0_k \right)&= \sum_{i,j=0}^N K^i_j a^0_k \otimes a^i_j + \sum_{i,j=0}^N K^i_j \otimes a^i_j a^0_k\\
&= \sum_{i,j=0}^N K^i_j a^0_k \otimes a^i_j + \sum_{j=0}^N K^k_j \otimes  a^0_j
\end{align*}
and
\begin{align*}
S_k^* K=  \left( a^0_k\otimes I + I \otimes a^0_k \right) \left(\sum_{i,j=0}^N K^i_j \otimes a^i_j \right)&= \sum_{i,j=0}^N a^0_k K^i_j \otimes a^i_j + \sum_{i,j=0}^N K^i_j \otimes a^0_k a^i_j  \\
&=\sum_{i,j=0}^N a^0_k K^i_j \otimes a^i_j  + \sum_{i=0}^N K^i_0 \otimes  a^i_k\,.
\end{align*}
From these equalities and since the operators $(a^i_j)_{i,j=0, ...,N}$ form a basis, the following system of equations is obtained
for $k= 1, \ldots, N$,
\begin{align*}
K^0_0 a^k_0 &= a^k_0 K^0_0 + K^0_k\\
a^0_k K^0_0&= K^0_0 a^0_k + K^k_0\,,
\end{align*}
for $j,k,l= 1, \ldots, N$ with $k \neq j$
\begin{align*}
K^0_j a^l_0 &= a^l_0 K^0_j\\
a^0_k K^0_j&= K^0_j a^0_k + K^k_j\\
K^0_j a^0_j + K^j_j&= a^0_j K^0_j + K^0_0\,,
\end{align*}
and
\begin{align*}
K_0^j a_l^0 &= a_l^0 K_0^j\\
 K_0^j a_0^k &= a_0^k K_0^j+ K^j_k\\
K_0^j a_0^j + K^0_0&= a^j_0 K_0^j + K^j_j\,,
\end{align*}
and for $i,j,k,l=1,\ldots, N$ with $k \neq i$ and $l \neq j$
\begin{align*}
K_j^i a^k_0 &= a^k_0 K_j^i\\
K^i_j a_l^0 &= a_l^0 K^i_j\\
K_j^i a^i_0 + K^0_j &= a^i_0 K_j^i\\
K_j^i a^0_j &= a^0_j K_j^i + K^i_0\,.
\end{align*}
We now study all these equations in order to prove that the $K^i_j$'s are all equal to $0$.
Note that the commutation of a matrix $M=(m_{ij})_{i,j=0,\ldots,N}$ with $a^k_0$ for $k=1,\ldots,N$ implies that for all $p\geq 1$ and $q=0,\ldots,N$ with $q\neq k$
$$ 
m_{00} = m_{kk} \ ,\quad m_{p,0} = 0 \quad \text{ and } \quad m_{k,q}=0\,.
$$
The commutation of $M$ with $a^0_k$ gives that for all $p\geq 1$ and $q=0,\ldots,N$ with $q\neq k$
$$ 
m_{00} = m_{kk} \ ,\quad m_{0,p} = 0 \quad \text{ and } \quad m_{q,k}=0\,.
$$
Thus since $K^0_j a^l_0 = a^l_0 K^0_j$ for all $j,l=1,\ldots, N$, the matrices $K^0_j$ are of the form
$$
\left(\begin{array}{cccc}
m_{00}& m_{01}&\hdots&m_{0N}\\
0&m_{00}& 0 &0 \\
\vdots & \ddots &\ddots & 0 \\
0 & \hdots& 0&m_{00}
 \end{array}\right)\,.
 $$
By the same way, from $K_0^j a_l^0 = a_l^0 K_0^j$ for all $j,l=1,\ldots, N$, we deduce that the matrices $K^j_0$ are of the form
$$
\left(\begin{array}{cccc}
m_{00}& 0&\hdots&0\\
m_{10}&m_{00}& 0 &0 \\
\vdots & \ddots &\ddots & 0 \\
m_{N0} & \hdots& 0&m_{00}
 \end{array}\right)\,.
 $$
Consider now the equations on $K^i_j$ for $i,j \neq 0$. Since we get $K_j^i a^k_0 = a^k_0 K_j^i$ and $K^i_j a_l^0 = a_l^0 K^i_j$ for $k,l= 1,\ldots, N$ with $k \neq i$ and $l \neq j$, the matrix $K^i_j$ is a diagonal matrix whose coefficients are all equal to $m_{00}$
except the column $j$ and the row $i$ with for the moment only zero coefficients on the first row and the first column.

In the following, the coefficients of a matrix $K^i_j$ are denoted by $(m^{ij}_{kl})_{k,l=0, \ldots,N}$.
We work then on the equation $K_j^i a^i_0 + K^0_j = a^i_0 K_j^i$. This equality gives that the diagonal coefficients of $K^0_j$ are $0$ and for $l=1, \ldots, N$ with $l\neq i$,
$$m_{00}^{ij}= m^{ij}_{ii}+ m^{0j}_{0i} \quad \text{ and }\quad m^{0j}_{0l}=-m^{ij}_{il}\,.$$
Then from $K_j^i a^0_j = a^0_j K_j^i + K^i_0$  we deduce that the diagonal coefficients of $K^i_0$ are $0$ and, for $l=1, \ldots, N$ with $l\neq j$,
$$m_{00}^{ij}= m^{ij}_{jj}+ m^{i0}_{0j} \quad \text{ and }\quad m^{i0}_{l0}=-m^{ij}_{lj}\,.$$
From the equalities $ K_0^j a_0^k = a_0^k K_0^j+ K^j_k$ and $a^0_k K^0_j= K^0_j a^0_k + K^k_j$ with $k\neq j$, we finally obtain that all the matrices $K^j_0$, $K^0_j$ and $K^j_k$ are all equal to $0$ for $j\neq k$. For $j=k$, the equalities 
$K_0^j a_0^j + K^0_0= a^j_0 K_0^j + K^j_j$ allow us to conclude that the only non zero operators are the $K^j_j$'s for $j=0,\ldots, N$ and all equal to $m^{00}_{00} I$.

Hence, we have proved that the commutant $\left\{ S_k, S_k^* ;   k=1, \ldots, N\right\}'$ is reduced to the operators of the form $\lambda I \otimes I$ with $\lambda$ in $\CC$. Then the commutant 
$$\left\{ H^A\otimes I + I \otimes H^B  + \dfrac{i}2  \sum_{j=1}^N (\beta_j- \beta_0)  (a^j_0 \otimes a^0_j - a^0_j \otimes a^j_0), S_k, S_k^* ;   k=1, \ldots, N\right\}'$$ is by definition a subset of $\left\{ S_k, S_k^* ;   k=1, \ldots, N\right\}'$. Therefore it is trivial too. This proves that the system has the property of return to equilibrium by the result in \cite{F-R}. 
\end{proof}

Since this state $\rho_{\beta}$ is the invariant state of $\mathcal H^A_S\otimes \mathcal H^B_S$, one deduces the thermalization of $\mathcal H^A_S$ and $\mathcal H^B_S$ by the environment.

\end{document}